\RequirePackage{amsmath}
\documentclass[runningheads,a4paper,english]{llncs}
\usepackage{upquote}
\usepackage{orcidlink}

\usepackage[ngerman,main=english]{babel}
\addto\extrasenglish{\languageshorthands{ngerman}\useshorthands{"}}
\usepackage{regexpatch}
\makeatletter
\edef\switcht@albion{%
  \relax\unexpanded\expandafter{\switcht@albion}%
}
\xpatchcmd*{\switcht@albion}{ \def}{\def}{}{}
\xpatchcmd{\switcht@albion}{\relax}{}{}{}
\edef\switcht@deutsch{%
  \relax\unexpanded\expandafter{\switcht@deutsch}%
}
\xpatchcmd*{\switcht@deutsch}{ \def}{\def}{}{}
\xpatchcmd{\switcht@deutsch}{\relax}{}{}{}
\edef\switcht@francais{%
  \relax\unexpanded\expandafter{\switcht@francais}%
}
\xpatchcmd*{\switcht@francais}{ \def}{\def}{}{}
\xpatchcmd{\switcht@francais}{\relax}{}{}{}
\makeatother

\makeatletter
\g@addto@macro{\UrlBreaks}{\UrlOrds}
\makeatother

\usepackage[%
    rm={oldstyle=false,proportional=true},%
    sf={oldstyle=false,proportional=true},%
    tt={oldstyle=false,proportional=true,variable=false},%
    qt=false%
]{cfr-lm}

\usepackage[T1]{fontenc}

\usepackage[
  babel=true, % Enable language-specific kerning. Take language-settings from the languge of the current document (see Section 6 of microtype.pdf)
  expansion=alltext,
  protrusion=alltext-nott, % Ensure that at listings, there is no change at the margin of the listing
  final % Always enable microtype, even if in draft mode. This helps finding bad boxes quickly.
]{microtype}

\DisableLigatures{encoding = T1, family = tt* }

\usepackage{graphicx}

\usepackage{diagbox}

\usepackage{xcolor}

\usepackage{listings}

\definecolor{eclipseStrings}{RGB}{42,0.0,255}
\definecolor{eclipseKeywords}{RGB}{127,0,85}
\colorlet{numb}{magenta!60!black}

\lstdefinelanguage{json}{
    basicstyle=\normalfont\ttfamily,
    commentstyle=\color{eclipseStrings}, % style of comment
    stringstyle=\color{eclipseKeywords}, % style of strings
    numbers=left,
    numberstyle=\scriptsize,
    stepnumber=1,
    numbersep=8pt,
    showstringspaces=false,
    breaklines=true,
    frame=lines,
    string=[s]{"}{"},
    comment=[l]{:\ "},
    morecomment=[l]{:"},
    literate=
        *{0}{{{\color{numb}0}}}{1}
         {1}{{{\color{numb}1}}}{1}
         {2}{{{\color{numb}2}}}{1}
         {3}{{{\color{numb}3}}}{1}
         {4}{{{\color{numb}4}}}{1}
         {5}{{{\color{numb}5}}}{1}
         {6}{{{\color{numb}6}}}{1}
         {7}{{{\color{numb}7}}}{1}
         {8}{{{\color{numb}8}}}{1}
         {9}{{{\color{numb}9}}}{1}
}

\usepackage[autostyle=true]{csquotes}

\defineshorthand{"`}{\openautoquote}
\defineshorthand{"'}{\closeautoquote}

\usepackage{booktabs}

\usepackage{paralist}

\usepackage[%
  square,        % for square brackets
  comma,         % use commas as separators
  numbers,       % for numerical citations;
  sort&compress  % as sort but in addition multiple numerical citations
]{natbib}

\usepackage{etoolbox}
\makeatletter
\patchcmd{\NAT@test}{\else \NAT@nm}{\else \NAT@hyper@{\NAT@nm}}{}{}
\makeatother

\SetExpansion
[ context = sloppy,
  stretch = 30,
  shrink = 60,
  step = 5 ]
{ encoding = {OT1,T1,TS1} }
{ }

\usepackage[rflt]{floatflt}

\usepackage{pdfcomment}

\usepackage{stfloats}
\fnbelowfloat

\usepackage[group-minimum-digits=4,per-mode=fraction]{siunitx}

\usepackage{hyperref}

\hypersetup{
  hidelinks,
  colorlinks=true,
  allcolors=black,
  pdfstartview=Fit,
  breaklinks=true
}

\usepackage[all]{hypcap}

\usepackage[caption=false,font=footnotesize]{subfig}

\usepackage[capitalise,nameinlink]{cleveref}

\crefname{section}{Sect.}{Sect.}
\Crefname{section}{Section}{Sections}
\crefname{listing}{List.}{List.}
\crefname{listing}{Listing}{Listings}
\Crefname{listing}{Listing}{Listings}
\crefname{lstlisting}{Listing}{Listings}
\Crefname{lstlisting}{Listing}{Listings}

\usepackage{lipsum}

\usepackage[math]{blindtext}
\usepackage{mwe}
\usepackage[realmainfile]{currfile}
\usepackage{tcolorbox}
\tcbuselibrary{listings}

\DeclareFontFamily{U}{MnSymbolC}{}
\DeclareSymbolFont{MnSyC}{U}{MnSymbolC}{m}{n}
\DeclareFontShape{U}{MnSymbolC}{m}{n}{
  <-6>    MnSymbolC5
  <6-7>   MnSymbolC6
  <7-8>   MnSymbolC7
  <8-9>   MnSymbolC8
  <9-10>  MnSymbolC9
  <10-12> MnSymbolC10
  <12->   MnSymbolC12%
}{}
\DeclareMathSymbol{\powerset}{\mathord}{MnSyC}{180}

\usepackage{xspace}

\makeatletter
\newcommand{\hydash}{\penalty\@M-\hskip\z@skip}
\makeatother

\input glyphtounicode
\usepackage{amssymb}
\usepackage{mathtools}
\usepackage{dsfont}

\usepackage{imakeidx}
\usepackage{tikz}
\usetikzlibrary{matrix}
\usepackage{bbold}
\usepackage{appendix}

\newcommand{\F}{\mathbb{F}}

\newcommand{\one}{\mathbb{1}}
\newcommand{\zero}{\mathbb{0}}

\newcommand{\eps}{\epsilon}
\newcommand{\sg}{\sigma}

\newcommand{\lb}{\lambda}

\DeclareMathOperator{\gl}{GL}
\DeclareMathOperator{\GL}{GL}
\DeclareMathOperator{\agl}{AGL}
\DeclareMathOperator{\AGL}{AGL}

\DeclareMathOperator{\sym}{Sym}

\DeclareMathOperator{\spn}{Span}
\newcommand{\defeq}{\vcentcolon=}
\newcommand{\ind}[1]{\mathds{1}_{#1}}

\renewcommand{\b}{\textbf{b}}

\newcommand{\Wo}{W_\circ}
\newcommand{\wo}{W_\circ}
\newcommand{\Uo}{U_\circ}
\newcommand{\uo}{U_\circ}
\newcommand{\Ho}{H_\circ}
\newcommand{\ho}{H_\circ}

\newcommand{\subs}{\subset}
\newcommand{\Sum}{\sum\limits}

\newcommand{\hex}[1]{\texttt{#1}$_\text{x}$}
\begin{document}

\title{Optimal s-boxes against alternative operations and linear propagation}

\author{Marco Calderini\orcidlink{0000-0002-6817-3421}\inst{1} \and
Roberto Civino\orcidlink{0000-0003-3672-8485}\inst{2} \and
Riccardo Invernizzi\orcidlink{0000-0002-2271-6822}\inst{3}}

\authorrunning{M. Calderini et al.}
\institute{University of Trento, \email{marco.calderini@unitn.it} \and
University of L'Aquila, 
\email{roberto.civino@univaq.it} \and
COSIC, KU Leuven, \email{riccardo.invernizzi@esat.kuleuven.be}}

\maketitle

\begin{abstract}
Civino et al.~(2019) have shown how some diffusion layers can expose a Substitution-Permutation Network
to vulnerability from differential cryptanalysis when employing alternative operations coming from groups isomorphic to the translation group on the message space. In this study, we present a classification of diffusion layers that exhibit linearity with respect to certain \emph{parallel} alternative operations, enabling the possibility of an alternative differential attack simultaneously targeting all the s-boxes within the block.
Furthermore, we investigate the differential behaviour with respect to alternative operations for all classes of optimal 4-bit s-boxes, as defined by Leander and Poschmann (2007). Our examination reveals that certain classes contain weak permutations w.r.t.\ alternative differential attacks. 
Finally, we leverage these vulnerabilities to execute a series of experiments showing the effectiveness of the cryptanalysis performed with a parallel alternative operation compared to the classical one.\\

\textbf{AMS classification 2020}: 20B35, 94A60,68P25.

\end{abstract}

\section{Introduction and preliminaries}
Differential cryptanalysis, originally introduced by Biham and Shamir in the late 1980s~\cite{biham1991differential} and subsequently generalised~\cite{W99,K94,BCJW02,BBS99}, has become one of the cornerstones for evaluating the robustness of various symmetric primitives. The fundamental premise of differential cryptanalysis is that analysing the differences (differentials) between pairs of plaintexts and the corresponding ciphertexts can unveil undesired biases.
While  differentials can be calculated with respect to any difference operator, regardless of which operation is responsible for performing the sum with the round key during encryption, it is usual for the two operations to coincide.
For this reason, classical differential cryptanalysis of a cipher in which the key is xor-ed to the state is typically performed by studying the distribution of xor-differentials,
whose propagation is traditionally prevented by the combined action of the linear diffusion layer and the s-box layer. In particular,
s-boxes are pivotal for ensuring the security of almost all contemporary block ciphers, serving as the primary nonlinear component within the cipher, particularly in the case of Substitution-Permutation Networks (SPNs). Equally relevant, the efficiency of a cipher is significantly influenced by the size of the s-boxes. In practical scenarios, s-boxes typically have a size of 4 or 8 bits, with 4 being the most popular choice for ciphers designed to operate on power-constrained devices~\cite{biham1998serpent,bogdanov2007present,shibutani2011piccolo,banik2015midori}. It is clear that the selection of appropriate s-boxes is critical to fortify the cipher against various types of attacks. In this sense, Leander and Poschmann have classified 4-bit s-boxes which are optimal w.r.t.\ standard criteria that guarantee poor propagation of xor-differentials~\cite{leander2007classification}.

A recent line of research is focused on the study of alternative operations for the differential cryptanalysis of xor-based ciphers~\cite{civino2019differential,calderini2021properties,tecseleanu2022security,nexexp,baudrin2023commutative}. These new operations define new differentials to be analysed. Within this approach, a large class of possible alternative operations has been studied, all of which have in common that they are induced by a group of translations isomorphic to the group of translations acting on the message space by means of the xor addition with the key.
In the context of an SPN, where the encrypted message is generated by iterating through a sequence of s-box layers, (xor)-linear diffusion, and xor-based key addition layers, altering the differential operator yields a dual impact. On one hand, it is highly probable that differentials cross the s-box layer more effectively, given that its nonlinearity is maximised with respect to xor. On the other hand, differentials do not deterministically propagate through the diffusion layer, as observed in classical scenarios. This pivotal limitation effectively restricts the success of the attack only to cases where the target layer is linear not only concerning xor but also with respect to the operation under consideration for computing differentials.
Moreover only a portion of round keys, called \emph{weak}, will let alternative differentials propagate deterministically through the xor-based key-addition layer, while all others  will require a probabilistic analysis.

\subsubsection*{Related works.}
The approach of differential cryptanalysis with alternative operations can be framed within the broader context of commutative cryptanalysis, introduced by Wagner~\cite{wagner2004towards}, of which differential cryptanalysis and some of its variations can be seen as particular cases.\\

A first successful attempt based on the study of the alternative differential properties of an xor-based toy cipher of the SPN family has shown that it is possible to highlight a bias in the distribution of the differences calculated compared to an alternative operation  which is instead not detectable by means of the standard xor-differential-based approach~\cite{civino2019differential}. The target cipher featured five 3-bit s-boxes and the operation used to perform the attack acted as the xor on the last four s-boxes, while on the first one matched with one of
the alternative sums defined by Calderini et al.~\cite{calderini2021properties}, coming from another translation group. The advantage of employing an alternative operation in this case was only derived from the benefit induced by a single s-box.
In a more recent experimental approach~\cite{nexexp}, we showed that better results in a similar context can be obtained using an \emph{alternative parallel operation}, in which every s-box can be targeted. In this case, the diffusion layer of the cipher was determined through an algorithm, ensuring that it adheres to the constraint of linearity with respect to both xor and the target operation.

The experiments mentioned above rely on the assumption that the translation group $T_\circ$ defining the new operation $\circ$ consists of functions that are affine with respect to the classical xor operation, that is denoted in this paper simply by $+$. In other words, each translation of the new operation can be expressed as an affine transformation over the usual vector space structure defined by xor. We also assume the dual condition, namely that the classical translations are affine with respect to $\circ$, which ensures that the weak-key space has a controlled and predictable dimension.
A similar approach is used by Baudrin et al.~\cite{baudrin2023commutative} to find probability-one distinguishers in both a modified \texttt{Midori}~\cite{banik2015midori} with \emph{weakened} round constants and in \texttt{Scream}. 
Their analysis focuses on finding affine maps $A,B \in \AGL(+)$ such that for sufficiently many encryption function holds 
\[
E_k(A(x)) = B(E_k(x)),
\] 
where $A$ and $B$ are diagonal matrices with s-box-sized blocks. 

%Their analysis focuses on a specific case where the diffusion layer is represented by a diagonal matrix with s-box-sized blocks, allowing both xor-differentials and differential based on an alternative parallel operation to propagate simultaneously. However, this assumption is not exhaustive, as the diffusion layer could be structured differently while still achieving the same effect. This highlights that the study of the group of maps which are linear w.r.t.\ two operations of the mentioned type is of interest both for its algebraic properties and its applications to cryptanalysis.

\subsubsection*{Our contribution.}
This paper establishes a general result for SPNs, characterising xor-linear maps that are simultaneously linear with respect to a fixed parallel alternative operation of affine type (Sec.~\ref{sec:aut}). Specifically, we consider a set of affine maps ${\tau_a} \subset \AGL(+)$, represented by diagonal matrices with s-box-sized blocks, where $a$ ranges over the message space, and forming a translation group under an operation $\circ$. We derive conditions on the cipher’s diffusion layer to ensure that, for each $a$ and $b$,
\[
 E_k(\tau_a(x)) = \tau_b(E_k(x)),
\]  
which implies that
\[
E_k(x)\circ E_k(x\circ a) = b
\]
with high probability for a set of sufficiently many keys, i.e.\ weak keys.
This finding enables the execution of a differential attack wherein each s-box affected by a nontrivial differential contributes to the final differential probability with increased efficacy compared to the conventional xor differentials. Additionally, differentials propagate deterministically through the linear layer in this scenario.
Moreover, we examine all possible alternative operations on 4 bits and investigate the differential properties of the optimal 4-bit s-boxes classes, as defined in the work of Leander and Poschmann~\cite{leander2007classification} (a comparable methodology, albeit in the context of modular addition, was recently employed by~\citet{zajac2020cryptographic}). Our analysis shows that each class comprises potentially weak permutations (Sec.~\ref{sec:sboxes}). When coupled with a diffusion layer as described earlier, these permutations have the potential to render the cipher susceptible to differential attacks with alternative operations. To substantiate our findings, we conclude the paper by presenting experimental results on a family of toy SPNs (Sec.~\ref{sec:experiments}).

\subsection{Notation}
Let $V$ be an $n$-dimensional vector space over $\F_2$ which represents the message space. We write $V = V_1 \oplus V_2 \oplus \cdots \oplus V_b$, where each $V_j$ is isomorphic to a vector space $B$ such that $\dim(B) = s$ on which every s-box acts. Therefore we have $n=sb$. We denote by $\{e_i\}_{i=1}^n$ the canonical basis of $V$.
If $G$ is any finite group acting on $V$, for each $g \in G$ and $v \in V$ we denote the action of $g$ on $v$ as $vg$, i.e.\ we use postfix notation for every function evaluation. We denote by $\sym(V)$ the symmetric group acting on $V$, i.e.\ the group of all permutations on the message space, by $\GL(V,+)$ the group of linear transformations, and by $\AGL(V,+)$ the group of affine permutations.  
The identity matrix of size $l$ is denoted by $\one_l$ and the zero matrix of size $l \times h$ is denoted by $\zero_{l,h}$, or simply $\zero_l$ if $l=h$. 
We finally denote by $T_+$ the group of translations on $V$, i.e. \[T_+ \defeq \{ \sg_a:  x \mapsto x+a \mid a \in V\} < \sym(V).\] We remind that the translation $\sg_k$ acts on a vector $x$ in the same way the key-addition layer of an SPN acts xor-ing the round key $k$ to the message $x$, i.e. $x\sg_k = x + k$.   

\subsection{Preliminaries on alternative operations}
\label{sec:alt_op}

An alternative operation on $V$ can be defined given any 2-elementary abelian regular subgroup $T < \AGL(V, +)$, that we can write as $T = \{ \tau_a \mid a \in V \}$, where $\tau_a$ is the unique element in $T$ which maps $0$ into $a$. Consequently, for all $a, b \in \, $V, we can define $a \circ b \defeq a \tau_b$, resulting in $(V, \circ)$ forming an additive group. The operation $\circ$ induces a vector space structure on $V$, with the corresponding group of translation being $T_\circ = T$. Additionally, for each $a \in V$, there exists $M_a \in \GL(V,+)$ such that $\tau_a = M_a \sg_a$, meaning that for every $x \in V$, \[x \circ a = x \tau_a = x M_a + a.\] It is also assumed throughout that $T_+ < \AGL(V,\circ)$, where $\AGL(V,\circ)$ is the normaliser in $\sym(V)$ of $T_\circ$ (i.e.\ the group of affine permutations w.r.t.\ $\circ$). This crucial technical assumption renders the key-addition layer an affine operator concerning the new operation, enabling the prediction of how the key addition affects the differentials with a reasonable probability. Further details on this aspect, which may not be directly relevant to the scope of the current paper,  can be found in Civino et al.~\cite{civino2019differential}. 
    In this context, we define the \emph{weak keys subspace} as 
    $$W_\circ \defeq \{a \mid a \in V, \sg_a = \tau_a \} = \{ k \mid k \in V, \ \forall x \in V \ x \circ k = x + k \}.$$
$W_\circ$ is a vector subspace of both $(V, +)$ and $(V, \circ)$. 
%It is worth noticing that each $k \in W_\circ$ represents a key that can be added regardlss of the operation involved. In this sense, we address those vectors as weak keys. 
It is known~\cite{caranti2005abelian,calderini2021properties} that $W_\circ$ is nonempty and that
\begin{equation}\label{eq:bound}
2 - (n \bmod 2) \leq \dim(W_\circ) \leq n-2.
\end{equation}
Moreover, up to conjugation, we can always assume $W_\circ$ to be the span of the last $d$ canonical vectors of $V$~\cite{calderini2021properties}. This allows to represent the new sum in a canonical way: for each $a \in V$ there exists a matrix $E_a \in \F_2^{(n-d) \times d}$ such that
    \begin{equation}\label{eq:form}
    M_a = 
    \begin{pmatrix}
    \ind{n-d} & E_a \\
    \mathbb{0}_{d, n-d} & \ind{d}
    \end{pmatrix}.
    \end{equation}
Fixing such an operation as above is therefore equivalent to defining the matrices
\[ 
M_{e_i} =
\begin{pmatrix}
\ind{n-d} & E_{e_i} \\
\mathbb{0}_{d, n-d} & \ind{d}
\end{pmatrix}
=
\left(\begin{array}{@{}c|c@{}}
   	\ind{n-d} & 
   	\begin{matrix}
   	\b_{i,1} \\
   	\vdots \\
   	\b_{i, n-d}
   	\end{matrix} \\
   	\hline
   	\mathbb{0}_{d, n-d} & \ind{d}
\end{array}\right)
\]
for $1 \leq i \leq n$, where $\b_{i,j} \in \F_2^d$. The assumptions on $T_\circ$ and on $W_\circ$ imply that $E_{e_i} = 0$ for $n-d+1 \leq i \leq n$, $\b_{i,i} = \textbf{0}$ and $\b_{i,j} = \b_{j,i}$. In conclusion, the following result characterises the criteria that the vectors $\b_{i,j}$ must adhere to in order to define an alternative operation as previously described.
\begin{theorem}[\cite{civino2019differential}]
	\label{thm:theta}
	Let $T_\circ < \AGL(V, +)$ be 2-elementary, abelian, and regular, and let $d \leq n-2$. The operation $\circ$ induced by $T_\circ$ is such that $d = \dim(\wo)$, $T_+ < \agl(V, \circ)$, and $W_\circ = \spn\{e_{n-d+1},\dots,e_n \}$ 
	if and only if the  matrix $\Theta_\circ \in (\F_{2^d})^{(n-d)\times (n-d)}$ defined as 
	\[
	\Theta_\circ \defeq 
	\begin{pmatrix}
		\mathbf{b}_{1,1} & \mathbf b_{1,2} & \cdots & \mathbf b_{n-d, 1} \\
		\mathbf b_{2,1} & \mathbf b_{2,2} & \cdots & \mathbf b_{n-d, 2} \\
		\vdots & \vdots & \ddots & \vdots \\
		\mathbf b_{n-d, 1} & \mathbf b_{n-d, 2} & \cdots & \mathbf b_{n-d, d}
	\end{pmatrix}
	\]
	is zero-diagonal, symmetric and no non-trivial $\mathbb F_2$-linear combination of its columns is the zero vector.
	The matrix $\Theta_\circ$ is also called the \emph{defining matrix} for the operation $\circ$.
\end{theorem}
This formulation also gives us an efficient algorithm to compute $a \circ b$ \cite{civino2019differential}. \\
%\begin{example}
%    Let $n = 3$ and $d = 1$. In light of the previous observation, we have only one possible alternative operation on $(\F_2)^3$, namely the operation $\diamond$ defined by
%    \[ 
%    M_{e_1} = 
%    \left(\begin{array}{@{}cc|c@{}}
%        1 & 0 & 0 \\
%        0 & 1 & 1 \\
%        \hline
%        0 & 0 & 1
%    \end{array}\right), \quad
%    M_{e_2} = 
%    \left(\begin{array}{@{}cc|c@{}}
%        1 & 0 & 1 \\
%        0 & 1 & 0 \\
%        \hline
%        0 & 0 & 1
%    \end{array}\right), \quad
%    M_{e_3} = \ind{3}
%    \]
%    In \textcolor{red}{the table}, $+$ and the operation $\diamond$ are compared. Each vector is interpreted as a binary number, most significant bit first, and then represented in hexadecimal notation. As $W_\diamond = \{0, e_3\}= \{\mathhex{0}, \mathhex{1}\}$, the first two rows and columns of the table are equal. Different entries are emphasised. 
%\end{example}
Every operation $\circ$ induces a \emph{product} $\cdot$ defined as follows:
\[
\forall a, b \in V \quad a\cdot b \defeq a + b + a \circ b.
\]
The product $a \cdot b$ can be interpreted as the error committed when confusing $a \circ b$ with $a + b$. This brings us to the definition of \emph{error space} of an operation $\circ$ as
\[
    \Uo \defeq \{a \cdot b :  a, b \in V\}.
\]
Finally, we recall some key properties of this product, which turns $(V,+,\cdot)$ into an alternating algebra of nilpotency class two. A comprehensive account of the underlying results can be found in Civino et al.~\cite{civino2019differential}.
\begin{proposition}
	\label{prop:info_prod}
    Let $\circ$ be an operation as above and let $\cdot$ be the induced product. Then
    \begin{itemize}
        \item $\cdot$ is distributive w.r.t\ $+$, i.e.\ $(V, +, \cdot)$ is an $\F_2$-algebra;
        \item $\Uo \subs \wo$;
        \item for each $x,y \in V$ there exists $\eps_{x,y} \in \uo$ such that
        \[
        x + y = x \circ y + \eps_{x, y},
        \]
        with $\eps_{x, y} = x \cdot y = (0, \dots, 0, (x_1,\dots, x_{n-d})E_y)$;
        \item $\uo$ is composed of all possible vectors $w \in \wo$ whose last $d$ components are $\F_2$-linear combinations of the defining vectors $\b_{ij}$;
        \item $x \cdot y \cdot z = 0$, for all $x, y, z \in V$.
    \end{itemize}
\end{proposition}

In the subsequent discussion, the term \emph{alternative operation} refers to an additive law $\circ$ on $V$ as defined above.

\section{Parallel operations and their automorphism groups}
\label{sec:aut}

Let $\circ$ be an alternative operation on the block-sized space $V$.  As outlined in the introduction, if $\lambda \in \GL(V,+)$ represents a (xor)-linear diffusion layer, and $\Delta \in V$ is an input difference traversing $\lambda$, predicting the output difference with respect to $\circ$, i.e.
\[
x\lambda \circ (x \circ \Delta)\lambda,
\]
becomes inherently challenging without additional assumptions on $\lambda$ that ensure a sufficiently high predictive probability.
For this reason, the examination of the following object becomes crucial: in cryptographic terms, it contains potential diffusion layers that allow differentials, whether computed with respect to xor or $\circ$, to propagate with a probability of 1.
\begin{definition}
Let $\circ$ be an alternative operation on $V$. Let us define 
\[
H_\circ := \{f \in \GL(V,+) \mid  \forall a,b \in V: \,\,\,(a\circ b)f=af\circ bf\}
\]
to be the subgroup of $\GL(V,+)$ of permutations that are linear w.r.t.\ the operation $\circ$. More precisely, denoting by $\AGL(V,\circ)$ the normaliser in $\sym(V)$ of $T_\circ$ and by $\GL(V,\circ)$ the stabiliser of $0$ in $\AGL(V,\circ)$, we have $H_\circ$ = $\GL(V,+)\cap \GL(V,\circ)$.
\end{definition}

\begin{remark}
Let us note that the mappings in $H_\circ$ are the automorphisms of the algebra $(V,+,\cdot)$ as in Proposition \ref{prop:info_prod}. Therefore, checking whether $\lambda\in\GL(V,+)$ is also in $\GL(V,\circ)$ is equivalent to verifying that $(a\cdot b) \lambda=a\lambda\cdot b\lambda$ for any $a,b\in V$.
\end{remark}

\begin{lemma}
	\label{lem:spazi_fissi}
	For each $\lb \in \Ho$ it holds $\Wo\lb = \Wo$ and $\Uo\lb = \Uo$.
\end{lemma}
\begin{proof}
	Let $\lambda \in \Ho$, and we begin by proving that $\Wo\lambda = \Wo$. Take any $a \in \Wo$ and $b \in V$. We aim to show that $a\lambda \circ b = a\lambda + b$. Now,
$$ a\lambda \circ b\lambda = (a \circ b)\lambda = (a + b)\lambda = a\lambda + b\lambda $$
and since $\lambda$ is invertible, it follows that $\Wo\lambda = \Wo$. \
On the other hand, if $a \in \uo$, then $a = b \cdot c$ for some $b, c \in V$. Therefore, $a\lambda = (b \cdot c)\lambda = b\lambda \cdot c\lambda$, which implies $a\lambda \in \uo$.\qed
\end{proof}

The structure of the group $H_\circ$ in its most general case has not been understood yet. This work addresses this challenge in a specific scenario, guided by assumptions that are deemed reasonable within the context of differential cryptanalysis.
\subsubsection*{Assumption 1: $\circ$ is a parallel operation.}

While the operation $\circ$ could, in theory, be defined on the entire message space $V$, studying the differential properties of the s-box layer, considered as a function with $2^n$ inputs, is impractical for standard-size ciphers. For this reason, we focus on operations applied in a \emph{parallel} way to each s-box-sized block, i.e., $\circ = (\circ_{1}, \circ_{2}, \dots, \circ_{b})$, where for each $1 \leq j \leq b$, $\circ_{j}$ is an operation on $V_j$. In this scenario, every operation is acting independently on the s-box space $B$, regardless of the others. This motivates the following definition.

\begin{definition}
Let $\circ$ be an alternative operation on $V$. We say that $\circ$ is \emph{parallel} if for each $1 \leq j \leq b$ there exists an alternative operation $\circ_j$ on $V_j$ such that for each $x,y \in V$ we have
\[
x \circ y = 
    \begin{pmatrix}
    	x_1 \\ \vdots \\ x_b
    \end{pmatrix}
    \circ
    \begin{pmatrix}
    	y_1 \\ \vdots \\ y_b
    \end{pmatrix}
    =
    \begin{pmatrix}
    	x_1 \circ_1 y_1 \\ \vdots \\ x_b \circ_b y_b
    \end{pmatrix},
    \]
    where $x = (x_1, x_2, \dots, x_b), y = (y_1, y_2, \dots, y_b)$ and each component belongs to the s-box-sized space, i.e. $x_j, y_j \in V_j \cong B$ for $1 \leq j \leq b$.
\end{definition}

In the notation of Sec.~\ref{sec:alt_op}, up to a block matrix conjugation, we can assume that every element $x \in V$ is associated to a translation $\tau_x = M_x \sg_x$, with 
\[
M_x = 
\begin{pmatrix}
	M_{x_1}^{\circ_1} & \cdots & \mathbb{0} \\
	\vdots & \ddots & \vdots \\
	\mathbb{0} & \cdots & M_{x_b}^{\circ_b} 
\end{pmatrix}
\]
where $M_{x_i}^{\circ_i}$ is the matrix associated to the translation $\tau_{x_i}$ with respect to the sum $\circ_i$, as defined in Eq.~\eqref{eq:form}. 

\subsubsection*{Assumption 2: $\dim(W_{\circ_j}) = s-2$.}
According to Eq.~\eqref{eq:bound}, every operation $\circ_j$ defined at the s-box level must satisfy the bound $\dim(W_{\circ_j}) \leq s-2$, being $s = \dim(B)$. 
The situation where the (upper) bound is reached holds particular interest for several reasons, as elaborated further in Civino et al.~\cite{civino2019differential}. Notably, 
\begin{itemize}
\item if the s-box size $s$ is four, the case where $\dim(W_{\circ_j}) = 2$ is the sole possibility;
\item imposing $\dim(W_{\circ_j}) = s - 2$ ensures optimal differential diffusion through the key addition layer (see e.g. \cite[Theorem 3.19]{civino2019differential});
\item other possible choices are at the moment less understood, and seem less suitable for our cryptanalytic application; the case $s-3$ is discussed as an example in Appendix \ref{sec:n3}.
\end{itemize}
For the reader's convenience, we present the classification result for $H_{\circ_j}$ obtained by Civino et al.\ in the considered case. Additionally, it is worth recalling that, according to Theorem \ref{thm:theta}, any $\circ_j$ for which $\dim(W_{\circ_j}) = s-2$ is determined by a single nonzero vector $\mathbf{b} \in (\F_2)^{s-2}$.

\begin{theorem}[\cite{civino2019differential}]
	\label{thm:caratt_ho_semplice}
	Let $\circ_j$ be an alternative operation such that $d= \dim(W_{\circ_j}) = s-2$ defined by a vector $\b \in (\F_2)^{s-2}$, and let $\lb \in (\F_2)^{s\times s}$. The following are equivalent:
	\begin{itemize}
		\item $\lb \in H_{\circ_j}$;
		\item there exist $A \in \GL((\F_2)^2, +)$, $D \in \GL((\F_2)^d, +)$, and $B \in (\F_2)^{2 \times d}$ such that
		\[
		\lb = 
		\begin{pmatrix}
			A & B \\
			\mathbb{0}_{d,2} & D
		\end{pmatrix}
		\]
		and $\b D = \b$.
	\end{itemize} 
\end{theorem}

We now recall the following result from Calderini et al.~\cite[Theorem 4.10]{calderini2021properties}.
\begin{theorem}
    \label{thm:coniugio_n2}
    Let $V = \F_2^s$, and $T$ and $T'$ be elementary abelian subgroups of $\agl(V, +)$ defining two operations $\circ$ and $\diamond$ respectively, such that $\dim(W_\circ) = \dim(W_\diamond) = s-2$. Then, there exists $g \in \gl(V)$ such that $T' = T^g = gTg^{-1}$.
\end{theorem}
This theorem can be easily extended to our parallel setting.
\begin{theorem}
	\label{thm:coniugio_parallelo}
	Let $\circ, \diamond$ be two parallel operation, defined by $\circ_1,...,\circ_b$ and $\diamond_1,...,\diamond_b$ respectively, such that for all $\circ_i, \diamond_i$ the weak key space has dimension $s-2$. Then, there exists $g \in \GL(V)$ such that $T_\diamond = T_\circ^g$. 
\end{theorem}
\begin{proof}
    For each $\circ_i$ and $\diamond_i$, Theorem \ref{thm:coniugio_n2} provides a matrix $g_i \in \gl(V_i)$ such that 
\[
T_{\diamond_i} = \{ M_{x_i}^{\diamond_i} \}_{x_i \in V_i} = g_i T_{\circ_i} g_i^{-1} = g_i \{ M_{x_i}^{\circ_i} \}_{x_i \in V_i} g_i^{-1}.
\]
Given the structure of the translations $\tau_x$, we can construct a matrix
\[
g = 
\begin{pmatrix}
   g_1 & \cdots & \mathbb{0} \\
   \vdots & \ddots & \vdots \\
   \mathbb{0} & \cdots & g_b
\end{pmatrix}.
\]
By construction, $g \in \gl(V)$. Additionally, Theorem \ref{thm:coniugio_n2} ensures that $T_\diamond$ is the conjugate of $T_\circ$ via $g$. This completes the proof.\qed
\end{proof}
\begin{lemma}
\label{lem:coniugio_ho}
    Let $T_\circ$ and $T_\diamond$ be elementary abelian regular subgroups of $AGL(V, +)$ defining two operations $\circ$ and $\diamond$ respectively, and such that $T_\diamond = T_\circ^g$ for $g \in \GL(V)$. Then $H_\diamond = H_\circ^g$.
\end{lemma}
\begin{proof}
    We know that $\agl(V, \circ)$ and $\agl(V, \diamond)$ are the normalisers of $T_\circ$ and $T_\diamond$, respectively. Since $T_\diamond = T_\circ^g$, it follows that $\agl(V, \diamond) = \agl(V, \circ)^g$. Consequently, we also have $\GL(V, \diamond) = \GL(V, \circ)^g$, where $\GL(V, \circ)$ is the stabiliser of $0$ in $\agl(V, \circ)$. Finally, the intersection with $\GL(V, +)$ remains preserved since $g \in \GL(V, +)$.\qed
\end{proof}

\begin{remark}
	\label{rmk:una_sola_somma}
	In the case $d = s-2$, Theorem \ref{thm:coniugio_parallelo} together with Lemma \ref{lem:coniugio_ho} allow us to restrict, up to conjugation, to the case $\circ_1 = ... = \circ_b$. For the sake of simplicity, and without loss of generality, we adopt this assumption from now on.
\end{remark}
We are now ready to present the characterisation of elements in the group $H_\circ$  for a parallel operation $\circ = (\circ_{1}, \circ_{2}, \dots, \circ_{b})$ with components at the s-box level satisfying $\dim(W_{\circ_j}) = s-2$. 
Notice that this more general result is valid for \emph{all} s-boxes in the relevant case  $s=4$ (see first comment in Assumption~2), but can also be valid for larger s-boxes. Therefore, the approach can be theoretically scaled to any SPN.
\begin{theorem}
	\label{thm:caratt_h0_msomme}
	Let $\circ =(\circ_{1}, \circ_{2}, \dots, \circ_{b})$ be a parallel alternative operation on $V$ such that for each $1\leq j \leq b$ $\circ_j$ is an alternative operation on $V_j$. Let us assume that every $\circ_{j}$ is such that $\dim(W_{\circ_j}) = s-2$ and it is defined by a vector $\b \in (\F_2)^{s-2}$.
	Let $\lb \in (\F_2)^{n\times n}$. Then, $\lb \in H_\circ$ if and only if it can be represented in the block form
	\[ 
	\lambda = 
	\left(\begin{array}{@{}c|c|c@{}}
		\begin{matrix}
			A_{11} & B_{11} \\
			C_{11} & D_{11}
		\end{matrix}
		&
		\cdots
		&
		\begin{matrix}
			A_{1b} & B_{1b} \\
			C_{1b} & D_{1b}
		\end{matrix}
		\\
		\hline
		\vdots
		&
		\ddots
		&
		\vdots
		\\
		\hline
		\begin{matrix}
			A_{b1} & B_{b1} \\
			C_{b1} & D_{b1}
		\end{matrix}
		&
		\cdots
		&
		\begin{matrix}
			A_{bb} & B_{bb} \\
			C_{bb} & D_{bb}
		\end{matrix}
		\\
	\end{array}\right),
	\]
	where
	\begin{enumerate}
		\item $A_{ij} \in (\F_2)^{2\times 2}$ such that for each row and each column of blocks there exists one and only one nonzero $A_{ij}$; moreover, all the nonzero $A_{ij}$ are invertible;
		\item $B_{ij} \in (\F_2)^{2 \times (s-2)}$;
		\item $C_{ij} = \mathbb{0}_{(s-2) \times 2}$;
		\item $D_{ij} \in (\F_2)^{(s-2) \times (s-2)}$ such that if $A_{ij}$ is zero, then $\b D_{ij}=\textbf{0}$, and if $A_{ij}$ is invertible, then $\b D_{ij} = \b$. Moreover, the matrix $D$ defined by
		\[
		D \defeq
		\begin{pmatrix}
			D_{11} & \cdots & D_{1b} \\
			\vdots & \ddots & \vdots \\
			D_{b1} & \cdots & D_{bb}
		\end{pmatrix}
		\]
		is invertible.
	\end{enumerate}
\end{theorem}
\begin{proof}

First, let us introduce some notation to simplify the proof. We denote \( e_i^j \) as the vector with a one in the \( i \)-th component of the \( j \)-th block, which is equivalent to \( e_{s(j-1) + i} \). Thus, the \( i \)-th strong component (the component spanned by nonweak keys in the \( i \)-th block) is \( \spn\{ e^i_1, e^i_2 \} \), and the \( i \)-th weak component (also denoted as \( \wo^i \)) is \( \spn\{ e^i_3, \dots, e^i_s \} \). We define \( \wo = \wo^1 \oplus \cdots \oplus \wo^b \).\\

\noindent (\(\Rightarrow\)) Let \( \lambda \in \ho \). We need to show that if \( \lambda \) is divided into blocks as described, all the blocks will have the desired properties. Notice that by definition, the block \( A_{ij} \) maps the \( i \)-th strong component to the \( j \)-th strong component, the block \( B_{ij} \) maps the \( i \)-th strong component to the \( j \)-th weak component, the block \( C_{ij} \) maps the \( i \)-th weak component to the \( j \)-th strong component, and the block \( D_{ij} \) maps the \( i \)-th weak component to the \( j \)-th weak component.

Firstly, by Lemma \ref{lem:spazi_fissi} (which extends naturally to the parallel setting), we have \( \wo \lambda = \wo \). This implies that \( C_{ij} = 0 \) for all \( i, j \). Consequently, \( \wo \) is contained in the image of \( D \); since \( \dim(\wo) = b(s-2) \), \( D \) must be invertible.

Now, let us focus on the blocks \( A_{ij} \). If we take a vector \( x_i \) with nonzero strong components in \( V_i \), we want to prove that \( x_i \lambda = x_j \) where \( x_j \) is a vector with nonzero strong components in \( V_j \) for some \( j \). We know that \( e^i_1 \cdot e^i_2 \neq 0 \); however, \( e^j_{i_1} \cdot e^k_{i_2} = 0 \) if \( j \neq k \), for all combinations of \( i_1, i_2 \in \{1, 2\} \).

Consider \( \spn \{e^i_1, e^i_2 \} \lambda \), and assume it contains an element with more than one nonzero strong component. Without loss of generality, suppose that the strong component of \( e_1^i \lambda\) is given by \(e^j_{r} + e^k_s \),  for some $k\ne j$ and \(r,s\in\{1,2\}\). Since we started with a 2-dimensional subspace, its image cannot span both strong subspaces of \( V_j \) and \( V_k \). If it does not generate \( V_k \), we would need another basis vector \( e^l_t \) with \( l \neq i \) such that \( e^l_t \lambda \) has a nonzero strong component in \( V_k \) different from \(e^k_s \). This implies
\[
0 = e_1^i \cdot e^l_t = (e_1^i \cdot e^l_t) \lambda = e_1^i \lambda \cdot e^l_t \lambda \neq 0,
\]
which is a contradiction. Therefore, each vector with only one nonzero strong component is mapped to a vector with only one nonzero strong component, meaning that for each row (or column) there is exactly one nonzero \( A_{ij} \). Furthermore, the nonzero \( A_{ij} \) must span the \( j \)-th strong subspace and must be invertible.

Since \( d = s-2 \), \( e^i_1 \cdot e^i_2 \) is the vector with \((0, 0, \b)\) in \( V_i \) and zero elsewhere. Therefore, the component in \( V_j \) of \( e^i_1 \cdot e^i_2 \) is \( \b D_{ij} \). If \( A_{ij} = 0 \), it must be that \( e^i_1 A_{ij} \cdot e^i_2 A_{ij} = 0 \), which implies \( \b D_{ij} = 0 \). Conversely, if \( A_{ij} \) is invertible, then \( \b D_{ij} = \b \) due to Lemma \ref{lem:spazi_fissi}.\\

\noindent(\(\Leftarrow\)) We now need to show that if \( \lambda \) has the form described in the statement, then \( \lambda \in \ho\). Since \( \lambda \in \gl(V,+) \) by construction, we only need to verify that \( (x \cdot y)\lambda = x\lambda \cdot y\lambda \) for all \( x, y \in V \). Moreover, since by Proposition \ref{prop:info_prod} the dot product is distributive with respect to \( + \), we can check this equality component-wise.

If either \( x \) or \( y \) is a weak key, both sides of the equation are zero. If \( x \) is a weak key, then \( x \cdot y = 0 \) by definition of the dot product. On the right-hand side, since all \( C_{ij} = 0 \), \( x\lambda \in \wo \), and thus \( x\lambda \cdot y\lambda = 0 \).

We are left with two cases: strong vectors from different blocks and strong vectors from the same block. For the first case, consider \( (e^j_{i_1} \cdot e^k_{i_2})\lambda = e^j_{i_1} \lambda \cdot e^k_{i_2} \lambda \) for \( j \neq k \) and \( i_1, i_2 \in \{1, 2\} \). Since \( j \neq k \) implies \( e^j_{i_1} \cdot e^k_{i_2} = 0 \), the condition on the \( A \) blocks of \( \lambda \) ensures that the part of \( e^j_{i_1} \lambda \) and \( e^k_{i_2}\lambda \) that lies outside \( \wo \) ends up in different blocks. Thus, \( e^j_{i_1} \lambda \cdot e^k_{i_2}\lambda = 0 \). No restriction is imposed on the components of \( e^j_{i_1} \lambda \) and \( e^k_{i_2}\lambda \) in \( \wo \), as the \( B \) blocks of \( \lambda \) are unrestricted. However, these components do not affect the dot product.

For the second case, strong vectors from the same block, consider \( (e^j_{i_1} \cdot e^j_{i_2})\lambda = e^j_{i_1} \lambda \cdot e^j_{i_2} \lambda \). If \( i_1 = i_2 \), both sides are zero. Without loss of generality, assume \( i_1 = 1 \) and \( i_2 = 2 \). Then, \( e_1^j \cdot e_2^j \) is a vector with \( (0, 0, \b) \) in the components corresponding to \( V_j \), and zero elsewhere. For a given \( j \), there is exactly one index \( k \) such that \( \b D_{jk} = \b \), and for all \( l \neq k \), \( \b D_{jl} = 0 \). Thus, the left-hand side is a vector with \( (0, 0, \b) \) in the components of \( V_k \) and zero elsewhere. Since \( A_{jk} \) is invertible and \( A_{jl} = 0 \) for \( l \neq k \), \( e_1^j \lambda \) and \( e_2^j \lambda \) are distinct vectors in the subspace \( V_k \), and their dot product is a vector with \( (0, 0, \b) \) in the same component, proving the equality. This completes the proof.\qed

\end{proof}

\section{Differential properties of optimal s-boxes}
\label{sec:sboxes}

In this section, we examine the differential properties of all possible 4-bit permutations, with respect to all possible alternative operations defined as in Sec.~\ref{sec:alt_op}. In particular, we set $s=4$ and therefore consider $\mathbb B:= \mathbb F_2^4$. We begin by acknowledging that, despite the compact size of the space, the count of alternative operations on $\mathbb B$ is considerable: 
%\RI{poco sopra (subito dopo l'assumption $s-2$) diciamo che una somma è determinata da un singolo vettore $\b \in \F_2^{2-s}$, che per $s=4$ sono 3 vettori + l'identità, mentre qua diciamo che ce ne sono 105; vogliamo aggiungere un commento su come sono coniugate? cosa cambia se consideriamo la media solo su quelle 3?}

\begin{proposition}[\cite{calderini2021properties}]\label{prop:dim4alt}
    There exist 105 elementary abelian regular subgroups groups $T_\circ$ in $\mathrm{AGL}(\mathbb{F}_2^4,+)$. Furthermore,  each of them satisfies $T_+<\mathrm{AGL}(\mathbb{F}_2^4,\circ)$ and $\dim W_\circ=s-2=2$.
\end{proposition}

We recall that given a permutation $f\in \sym(\mathbb B)$ we can define \[\delta_f(a,b)=\#\{x \in B\mid xf+(x+a)f=b\}.\] 

The \emph{differential uniformity} of $f$ is defined as $\delta_f:=\max_{a\ne 0}\delta_f(a,b)$ and it represents the primary metric to consider when assessing the resistance of an s-box to differential cryptanalysis~\cite{nyberg1993differentially}.

Several cryptographic properties, including differential uniformity, are preserved under affine equivalence for vectorial Boolean functions \cite{carlet2021boolean}.
Two functions, denoted as $f$ and $g$, are considered \emph{affine equivalent} if there exist two affine permutations, $\alpha$ and $\beta$, in $\AGL(V,+)$ such that $g = \beta  f  \alpha$.

Leander and Poschmann~\cite{leander2007classification} provided a comprehensive classification (up to affine equivalence) of permutations over $\mathbb B= \mathbb F_2^4$. They identified 16 classes with \emph{optimal} cryptographic properties. All 16 classes exhibit a classical differential uniformity equal to 4, which represents the best possible value for s-boxes in $\sym(\mathbb B)$. The representatives of the 16 classes are listed in Table \ref{tab:rep}, where each vector is interpreted as a binary number, most significant bit first.

\begin{table}[]
    \label{tab:rep}
        \caption{Optimal 4-bit permutations according to Leander and Poschmann}
    \centering
     \begin{tabular}{l||cccccccccccccccc|}
	&{\hex{0}} &{\hex{1}} &{\hex{2}}& \hex{3}& \hex{4}& \hex{5}& \hex{6}& \hex{7}& \hex{8}& \hex{9}& \hex{A}& \hex{B}& \hex{C}& \hex{D}& \hex{E}& \hex{F}\\
	\hline\hline
 $G_0$ &\hex{0}&\hex{1}&\hex{2}&\hex{D}&\hex{4}&\hex{7}&\hex{F}&\hex{6}&\hex{8}&\hex{B}&\hex{C}&\hex{9}&\hex{3}&\hex{E}&\hex{A}&\hex{5}\\
 \hline
$G_1$ &\hex{0}&\hex{1}&\hex{2}&\hex{D}&\hex{4}&\hex{7}&\hex{F}&\hex{6}&\hex{8}&\hex{B}&\hex{E}&\hex{3}&\hex{5}&\hex{9}&\hex{A}&\hex{C}\\
\hline
$G_2$ &\hex{0}&\hex{1}&\hex{2}&\hex{D}&\hex{4}&\hex{7}&\hex{F}&\hex{6}&\hex{8}&\hex{B}&\hex{E}&\hex{3}&\hex{A}&\hex{C}&\hex{5}&\hex{9}\\\hline
$G_3 $&\hex{0}&\hex{1}&\hex{2}&\hex{D}&\hex{4}&\hex{7}&\hex{F}&\hex{6}&\hex{8}&\hex{C}&\hex{5}&\hex{3}&\hex{A}&\hex{E}&\hex{B}&\hex{9}\\\hline
$G_4$ &\hex{0}&\hex{1}&\hex{2}&\hex{D}&\hex{4}&\hex{7}&\hex{F}&\hex{6}&\hex{8}&\hex{C}&\hex{9}&\hex{B}&\hex{A}&\hex{E}&\hex{5}&\hex{3}\\\hline
$G_5$ &\hex{0}&\hex{1}&\hex{2}&\hex{D}&\hex{4}&\hex{7}&\hex{F}&\hex{6}&\hex{8}&\hex{C}&\hex{B}&\hex{9}&\hex{A}&\hex{E}&\hex{3}&\hex{5}\\\hline
$G_6$ &\hex{0}&\hex{1}&\hex{2}&\hex{D}&\hex{4}&\hex{7}&\hex{F}&\hex{6}&\hex{8}&\hex{C}&\hex{B}&\hex{9}&\hex{A}&\hex{E}&\hex{5}&\hex{3}\\\hline
$G_7 $&\hex{0}&\hex{1}&\hex{2}&\hex{D}&\hex{4}&\hex{7}&\hex{F}&\hex{6}&\hex{8}&\hex{C}&\hex{E}&\hex{B}&\hex{A}&\hex{9}&\hex{3}&\hex{5}\\\hline
$G_8$ &\hex{0}&\hex{1}&\hex{2}&\hex{D}&\hex{4}&\hex{7}&\hex{F}&\hex{6}&\hex{8}&\hex{E}&\hex{9}&\hex{5}&\hex{A}&\hex{B}&\hex{3}&\hex{C}\\\hline
$G_9$ &\hex{0}&\hex{1}&\hex{2}&\hex{D}&\hex{4}&\hex{7}&\hex{F}&\hex{6}&\hex{8}&\hex{E}&\hex{B}&\hex{3}&\hex{5}&\hex{9}&\hex{A}&\hex{C}\\\hline
$G_{10}$&\hex{0}&\hex{1}&\hex{2}&\hex{D}&\hex{4}&\hex{7}&\hex{F}&\hex{6}&\hex{8}&\hex{E}&\hex{B}&\hex{5}&\hex{A}&\hex{9}&\hex{3}&\hex{C}\\\hline
$G_{11}$&\hex{0}&\hex{1}&\hex{2}&\hex{D}&\hex{4}&\hex{7}&\hex{F}&\hex{6}&\hex{8}&\hex{E}&\hex{B}&\hex{A}&\hex{5}&\hex{9}&\hex{C}&\hex{3}\\\hline
$G_{12}$&\hex{0}&\hex{1}&\hex{2}&\hex{D}&\hex{4}&\hex{7}&\hex{F}&\hex{6}&\hex{8}&\hex{E}&\hex{B}&\hex{A}&\hex{9}&\hex{3}&\hex{C}&\hex{5}\\\hline
$G_{13}$&\hex{0}&\hex{1}&\hex{2}&\hex{D}&\hex{4}&\hex{7}&\hex{F}&\hex{6}&\hex{8}&\hex{E}&\hex{C}&\hex{9}&\hex{5}&\hex{B}&\hex{A}&\hex{3}\\\hline
$G_{14}$&\hex{0}&\hex{1}&\hex{2}&\hex{D}&\hex{4}&\hex{7}&\hex{F}&\hex{6}&\hex{8}&\hex{E}&\hex{C}&\hex{B}&\hex{3}&\hex{9}&\hex{5}&\hex{A}\\\hline
$G_{15}$&\hex{0}&\hex{1}&\hex{2}&\hex{D}&\hex{4}&\hex{7}&\hex{F}&\hex{6}&\hex{8}&\hex{E}&\hex{C}&\hex{B}&\hex{9}&\hex{3}&\hex{A}&\hex{5}\\\hline
 	\end{tabular}
\end{table}

\subsection{Dealing with affine equivalence}
Our goal is to analyse the differential uniformity of each optimal s-box class, with respect to every alternative operation $\circ$ on $\mathbb B$. 
The definitions given above can be generalised in the obvious way setting
$\delta^\circ_f(a,b)=\#\{x \in \mathbb B \mid xf\circ (x\circ a)f=b\}$ and   calling \emph{$\circ$-differential uniformity} of $f$ the value
$\delta^\circ_f:=\max_{a\ne 0}\delta^\circ_f(a,b)$.

It is noteworthy that, unlike in the case of classic differential uniformity, the value of $\delta^\circ_f$ is not invariant under affine equivalence. However, verifying the $\circ$-differential uniformity of $g_2  G_i  g_1$ for any optimal class and every pair $g_1, g_2 \in \AGL(\mathbb B,+)$ would be impractical. 
Therefore, a reduction in the number of permutations to be checked is necessary, and for this purpose, we make the following observations.
First, similar to the classical case, the $\circ$-differential uniformity is preserved under affine transformations w.r.t.\ $\circ$.
\begin{proposition}
	\label{prop:diff_inv_ho}
 Given $f \in \sym(\mathbb B)$ and $g_1, g_2 \in  \AGL(\mathbb B, \circ)$ we have
	\[ \delta_{g_1  f  g_2}^{\circ}(a,b)= \delta_f^{\circ}(g_2(a),g_1^{-1}(b)). \]
\end{proposition}
Moreover, Proposition \ref{prop:dim4alt} establishes that for any $\circ$ derived from a translation group in $\AGL(\mathbb B,+)$, the $+$-translations are affine with respect to $\circ$.
This initial observation allows us to narrow down the analysis to $g_2  G_i  g_1$ with $g_1, g_2 \in \GL(\mathbb B,+)$, which still remains impractical. Furthermore, considering that $\ho = \gl(\mathbb B, +) \cap \GL(\mathbb B, \circ)$, Proposition \ref{prop:diff_inv_ho} establishes that left and right multiplication by elements in $\ho$ preserves both $\circ$ and $+$-differential uniformity. 
It is noteworthy that during this process, the rows of the matrix containing all the $\delta^\circ_f(a,b)$ ($\mathrm{DDT}^\circ$) are merely shuffled, thereby preserving the highest element of each row. 
Therefore, the following conclusion can be easily obtained.
\begin{proposition}
    Let $g_1, g_2 \in \gl(\mathbb B,+)$ and $f\in \sym(\mathbb B)$.  For any $g_1' \in g_1 H_\circ $ and $g_2' \in H_\circ g_2$ we have \[\delta^\circ_{g_2 f g_1}=\delta^\circ_{g_2' f g_1'}.\]
\end{proposition}
\begin{proof}
\setcounter{footnote}{0}
  Take $h_1, h_2\in H_\circ$ such that $g'_1=g_1 h_1$ and $g_2'=h_2 g_2$. Then\footnote{We remind here that we use postfix notation in this paper, so that $x f$ denotes the image of $x$ under the function $f$.},
    \[
   x g'_2 f g_1' \circ  (x \circ a)g'_2 f g_1'= xh_2  g_2 f g_1 \circ  (xh_2 \circ ah_2g_2 f g_1)h_1,
    \]
   implying that $\delta^\circ_{g'_2 f g_1'}(a,b)=\delta^\circ_{g_2 f g_1}(ah_2,bh_1^{-1})$. So, $\delta^\circ_{g'_2 f g_1'}=\delta^\circ_{g_2 f g_1}$.\qed
\end{proof}

The final proposition allows us to focus solely on $g_1$ and $g_2$ within the left and right cosets of $H_\circ$. These reductions facilitate the analysis of the potential $\circ$-differential uniformities attainable across all classes of optimal permutations for the 105 conceivable alternative sums defined over $\mathbb B$.
For each of the 105 alternative operations, we systematically explored each of the 16 classes, following the described procedure, and we recorded the $\circ$-differential uniformity for every candidate. To streamline the presentation, we calculated the average across all 105 operations and presented the consolidated results in Tab.~\ref{tab:differential_avg}.

\begin{table}[h!]\label{tab:differential_avg}
	\centering
	\caption{Avg. number of functions with given $\circ$-differential uniformity}
	\label{tbl:10}
		{\footnotesize\begin{tabular}{c||c|c|c|c|c|c|c|c|}
		 %\backslashbox{Class}{$\delta^\circ$}
		 Class $\downarrow$ // $\delta^\circ \rightarrow$& 2 & 4    & 6    & 8    & 10  & 12 & 14 & 16 \\ \hline\hline
		$G_0$    & 0& 780& 6695& 2956& 359& 16& 0& 12   \\
		$G_1$    & 0& 682& 6927& 2823& 374& 0& 0& 12 \\
		$G_2$    & 0& 781& 6695& 2956& 359& 16& 0& 12  \\
		$G_3$    & 0& 896& 7566& 2210& 146& 0& 0& 0  \\
		$G_4$    & 0& 1104& 7770& 1825& 118& 0& 0& 0   \\
		$G_5$    & 0& 822& 7994& 1790& 212& 0& 0& 0  \\
		$G_6$    & 0& 1120& 7441& 2108& 150& 0& 0& 0 \\
		$G_7$    & 0& 898& 7628& 2139& 133& 20& 0& 0  \\
		$G_8$    & 0& 859& 6503& 3102& 296& 48& 0& 12 \\
		$G_9$    & 0& 1123& 7062& 2457& 141& 36& 0& 0  \\
		$G_{10}$ & 0& 1084& 7115& 2437& 147& 36& 0& 0 \\
		$G_{11}$ & 0& 1202& 7299& 2159& 157& 0& 0& 0  \\
		$G_{12}$ & 0& 1099& 7275& 2291& 153& 0& 0& 0  \\
		$G_{13}$ & 0& 916& 7749& 1965& 176& 12& 0& 0   \\
		$G_{14}$ & 0& 1122& 7100& 2400& 149& 48& 0& 0  \\
		$G_{15}$ & 0& 1122& 7100& 2400& 149& 48& 0& 0 \\
  \hline
	\end{tabular}}
\end{table}

In our examination, we observe that if, for a given operation $\circ$, certain elements within an affine equivalence class yield a $\circ$-differential uniformity $\delta$, then this value $\delta$ is achieved by some element in the entire class for all alternative operations.
Our analysis reveals that certain optimal functions may exhibit the highest differential uniformity (16) for alternative operations, specifically the classes $G_0$ (containing, e.g., the s-box S1 of Serpent~\cite{biham1998serpent}), $G_1$ (containing, e.g., the s-box of PRESENT~\cite{bogdanov2007present} and the s-box S2 of Serpent), $G_2$ (containing, e.g., the s-box S0 of Serpent), and $G_8$ (containing, e.g., the s-box of NOEKEON \cite{noekeon}). Conversely, the classes $G_3$, $G_4$, $G_5$, $G_6$, $G_{11}$, and $G_{12}$ demonstrate more favorable behavior concerning alternative operations.

\begin{remark}
    We observe from Table \ref{tab:differential_avg} that we do not obtain maps with {\(\circ\)-differential} uniformity equal to 2 or 14. This is because, with respect to the xor operation, such differential uniformities are not present for permutations in 4 bits \cite{leander2007classification}. The alternative operations \(\circ\) that we are considering induce a vector space structure on \(V\) isomorphic to \((V, +)\). Specifically, let \(\phi: (V, \circ) \to (V, +)\) be a vector space isomorphism. For any map \(f: V \to V\), the map \(\phi  f  \phi^{-1}\) will have \(\circ\)-differential uniformity equal to the classical differential uniformity of \(f\). This is due to the following relationship:
\[
(x \circ a) \phi  f  \phi^{-1} \circ (x) \phi  f  \phi^{-1} = ((x \phi + a \phi) f + (x \phi) f) \phi^{-1}.
\]\end{remark}
\section{Experiments on a 16-bit block cipher with 4-bit s-boxes}
\label{sec:experiments}
\setcounter{footnote}{1}
In this concluding section, we aim to apply the results obtained above to a family of (toy) ciphers\footnote{The implementation of our simulations is made available here: \url{https://github.com/KULeuven-COSIC/alternative-differential-cryptanalysis}}. These ciphers may exhibit security under classical differential cryptanalysis but reveal vulnerabilities to the alternative differential approach.

In our experiments, we set $V = \F_2^{16}$, $n=4$, and $s=4$, defining $\circ$ as the parallel sum by applying the alternative operation defined by the vector $\mathbf{b} = (0, 1)$ to each $4$-bit block.
Moreover, all our ciphers will feature the 4-bit permutation $\gamma :\F_2^4 \rightarrow \F_2^4$ defined by the  sequence 
 (\hex{0}, \hex{E}, \hex{B}, \hex{1}, \hex{7}, \hex{C}, \hex{9}, \hex{6}, \hex{D}, \hex{3}, \hex{4}, \hex{F}, \hex{2}, \hex{8}, \hex{A}, \hex{5}) as its s-box.
Precisely, four copies of $\gamma$ will act on the 16-bit block. Notice that the s-box $\gamma \in $ belongs to $G_0$ and has 
$\delta_\gamma=4$ and $\delta^\circ_\gamma=16$.

In all the experiments described below, we consider the SPN whose $i$-th round
is obtained by the composition of the parallel application of the s-box $\gamma$ on every
4-bit block, a `diffusion layer' $\lambda$ sampled randomly from $H_\circ$, and the xor with the
$i$-th random round key. We study the difference propagation in the cipher in a
long-key scenario, i.e., the key-schedule selects a random long key $k \in \mathbb F_2^{16r}$ where
$r$ is the number of rounds. To avoid potential bias from a specific key choice,
we conduct our experiments by averaging over $2^{15}$ random long-key generations.
This approach gives us a reliable estimate of the expected differential probability
for the best differentials in this cipher.

In 150 distinct executions, spanning a range of rounds from 3 to 10, we calculated the discrepancy between the most effective $\circ$-trail and $+$-trail. To manage computational resources, our focus was narrowed down to input differences with a Hamming weight of 1.

The results are depicted in Fig.~\ref{fig:diffplot}, where each dot represents an individual simulation. The $x$ axis corresponds to the negative logarithm of the probability of the best $\circ$ differential, while the $y$ axis represents the difference between that value and the negative logarithm of the probability of the best $+$ differential. Darker dots indicate a higher number of rounds, as explained in the legend. Notably, about half of the dots lie above zero, suggesting that the best $\circ$ differential consistently outperforms the best $+$ differential until they become indistinguishable. Interestingly, this convergence often occurs when the $\circ$ probability is already very close to $2^{-16}$, providing potential candidates for our distinguisher attack. 

\begin{figure}[h!]
    \caption{Comparison of $+$ and $\circ$ trails for random mixing layers}
    \label{fig:diffplot}
    \centering
    \includegraphics[scale=0.70]{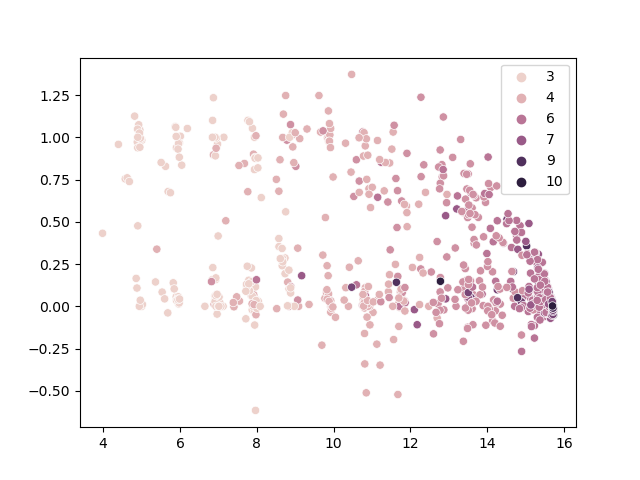}
\end{figure}

\subsubsection*{Acknowledgments.}
M. Calderini and R. Civino are members of INdAM-GNSAGA (Italy) and thankfully acknowledge support by MUR-Italy via PRIN 2022RFAZCJ `Algebraic methods in cryptanalysis'. 
R.\ Civino is supported by the Centre of EXcellence on Connected, Geo-Localized and 
 Cybersecure Vehicles (EX-Emerge), funded by Italian Government under CIPE resolution n.\ 70/2017 (Aug.\ 7, 2017). 
R.~Invernizzi is supported by the European Research Council (ERC) under the European Union’s Horizon 2020 research and innovation programme (grant agreement ISOCRYPT - No. 101020788) and by CyberSecurity Research Flanders with reference number VR20192203.\\
\subsubsection*{Author Contribution information.}
All authors contributed to writing this manuscript and all authors read and approved the submitted version.
\subsubsection*{Competing Interest information.}
The authors have no financial or proprietary interests in any material discussed in this article.

%%% ===============================================================================
%%% Bibliography
%%% ===============================================================================
%\newpage

\renewcommand{\bibsection}{\section*{References}} % requried for natbib to have "References" printed and as section*, not chapter*
% Use natbib compatbile splncs04nat style.
% It does provide all features of splncs04.bst, but is developed in a clean way.
% Source: https://github.com/tpavlic/splncs04nat

\bibliographystyle{splncs04nat}
\bibliography{refs}{}

%\bibliographystyle{splncs04nat}
%\begingroup
%  \microtypecontext{expansion=sloppy}
%  \small % ensure correct font size for the bibliography
%  \bibliography{paper}
%\endgroup

\newpage
\appendix

\section{The case $d = s-3$}
\label{sec:n3}
A natural next step is to extend the characterization of Theorem \ref{thm:caratt_ho_semplice} to higher-dimensional $\wo$. In this section we will consider a new operation $\circ$ such that $d = \dim(\Wo) = s-3$. Notice that the following results do not apply in the case $s=4$ and require larger values of the s-box size. Thanks to Theorem \ref{thm:theta}, its defining matrix can be written as 
\[
\Theta_\circ = 
\begin{pmatrix}
	\textbf{0} & \b_{21} & \b_{31} \\
	\b_{21} & \textbf{0} & \b_{32} \\
	\b_{31} & \b_{32} & \textbf{0}
\end{pmatrix}
\]
where no $\F_2$-linear combination of columns the null vector. The multiplication matrices are
\[ 
M_{e_1} = 
\left(\begin{array}{@{}c|c@{}}
	\ind{3} &
	\begin{matrix}
		\textbf{0} \\
		\b_{21} \\
		\b_{31}
	\end{matrix}
	\\
	\hline
	\mathbb{0}_{s-3, 3} &
	\ind{s-3}
\end{array}\right),
\ M_{e_2} = 
\left(\begin{array}{@{}c|c@{}}
	\ind{3} &
	\begin{matrix}
		\b_{21} \\
		\textbf{0} \\
		\b_{32}
	\end{matrix}
	\\
	\hline
	\mathbb{0}_{s-3, 3} &
	\ind{s-3}
\end{array}\right),
\ M_{e_3} = 
\left(\begin{array}{@{}c|c@{}}
	\ind{3} &
	\begin{matrix}
		\b_{31} \\
		\b_{32} \\
		\textbf{0}
	\end{matrix}
	\\
	\hline
	\mathbb{0}_{s-3, 3} &
	\ind{s-3}
\end{array}\right)
\]
while as always $M_{e_j} = \ind{s}$ for $j > 3$. In this case we have to deal with multiple error vectors, namely
\[
e_i \cdot e_j = e_j \cdot e_i = (0, 0, 0, \b_{ij}).
\]
It is hence convenient to introduce the notation $u_{ij} = u_{ji} \defeq e_i \cdot e_j = (0, 0, 0, \b_{ij})$.
\begin{proposition}
	Let $\circ$ be an operation such that $\dim(\wo) = s-3$. Then $\dim(\uo) \in \{2, 3\}$.
\end{proposition}
\begin{proof}
	Thanks to Proposition \ref{prop:info_prod}, $\uo$ is made of all the vectors $w \in \wo$ whose last $d$ components are all the possible $\F_2$-linear combinations of the vectors $\b_{ij}$. By Theorem \ref{thm:theta}, at least two of these vectors are granted to be independent, hence $\dim(\uo) \geq 2$. Moreover, they are all spanned by $u_{12},u_{13}$ and $u_{23}$, hence $\dim(\uo) \leq 3$.\qed
\end{proof}

\begin{theorem}
    \label{thm:caratt_ho_n3}
    Let $\circ$ be an operation such that $\dim(\wo) = s-3$, defined by $\Theta_\circ$ as above. Let $\lb \in (\F_2)^{s\times s}$. Then $\lb \in \ho$ if and only if we can write
    $$ 
    \lb = 
    \begin{pmatrix}
        A & B \\
        \mathbb{0}_{d,3} & D
    \end{pmatrix}
    $$
    for some $A \in \GL((\F_2)^3, +)$, $D \in \GL((\F_2)^d, +)$, and $B \in (\F_2)^{3 \times d}$ such that the three \emph{compatibility equations}
    \[
        \b_{ij}D = \Sum_{k_1, k_2 =1}^3 (A_{ik_1}A_{jk_2} + A_{ik_2}A_{jk_1})\b_{k_1k_2}
    \]		
    for $1 \leq i,j \leq 3$, $i \neq j$ are satisfied, where $A_{ij}$ is the entry of $A$ in the $i$-th row and $j$-th column.
\end{theorem}
\begin{proof}

Let us begin with $\lb \in \ho$. By Lemma \ref{lem:spazi_fissi}, the first three components of the last $d$ rows must be zero, as they would otherwise map elements from $\wo$ outside of $\wo$. Furthermore, $\lb \in \gl((\F_2)^s, +)$ implies that both $A \in \GL((\F_2)^3, +)$ and $D \in \GL((\F_2)^d, +)$ are invertible. For any $i, j$, we can express $u_{ij} \lb$ as $e_i \lb \cdot e_j\lb$. This gives us:
\[
e_i \lb = \Sum_{k=1}^3 A_{ik} e_k + \Sum_{k=1}^{n-3} B_{ik} e_{k+3}
\]
and similarly for $e_j$. Since the dot product is distributive over $+$, we can disregard the second sum, which consists only of weak vectors. Thus, we obtain:
\[
e_i \lb \cdot e_j \lb = \left(\Sum_{k_1=1}^3 A_{ik_1} e_{k_1}\right) \cdot \left(\Sum_{k_2=1}^3 A_{jk_2} e_{k_2}\right) = \Sum_{k_1, k_2=1}^3 \left( A_{ik_1}A_{jk_2} + A_{ik_2}A_{jk_1} \right) u_{ij}.
\]
Note that these equations reduce to zero when $i = j$; since $u_{ii} = 0$, we can focus on the case $i \neq j$. We know that the first three components of these products are zero, so we restrict our attention to the component in $\wo$ on both sides. The above expression simplifies to:
\[
\b_{ij}D = \Sum_{k_1, k_2=1}^3 (A_{ik_1} A_{jk_2} + A_{ik_2} A_{jk_1}) \b_{k_1 k_2}
\]
for $1 \leq i,j \leq 3$, $i \neq j$.

Now, suppose $\lb$ is in the form prescribed by the theorem; we aim to prove that $\lb \in \ho$. The zero block and the invertibility of $A$ and $D$ are sufficient to show that $\lb \in \gl(V, +)$. As in Theorem \ref{thm:caratt_h0_msomme}, it suffices to show that $(e_i \cdot e_j) \lb = e_i \lb \cdot e_j \lb$. If at least one of $e_i, e_j \in \wo$ (say $e_i$), then the left-hand side is zero. Furthermore, the zero block in the lower left of the matrix ensures that $e_i \lb \in \wo$, which implies the right-hand side is also zero. We are left with the cases where $e_i, e_j \not\in \wo$. If $i = j$, we have $0 = 0$ by the definition of the dot product. The three cases where $i \neq j$ correspond exactly to the three compatibility equations, as shown in the first part of the proof. \qed
\end{proof}

We will now present two examples of how Theorem \ref{thm:caratt_ho_n3} allows us to count the number of elements in $\ho$. These computations will be important in proving Theorem \ref{thm:coniugio_n3}.

\begin{example}
	\label{ex:n6_d3_u3}
		Let $V = (\F_2)^6$, and let $\circ$ be an operation where $\dim(\wo) = \dim(\uo) = s-3 = 3$. By applying Theorem \ref{thm:caratt_ho_n3}, we can determine the size of $\ho$. Specifically, $A$ can be any matrix from $\gl((\F_2)^3)$, providing 168 possible choices. Note that $A$ defines the image of $e_i$ for $1 \leq i \leq 3$, and thus determines the image of $\uo$. Since in this case $\uo = \wo$, the compatibility equations uniquely determine $D$ for each selected $A$ (which is not always the case: for each $A$, any $D$ that fixes $\uo$ and permutes the remaining vectors in a basis of $\wo$ can be chosen). Additionally, $B$ can be any $3 \times 3$ matrix, giving us 512 options. Therefore, the total size of $\ho$ is 86016.
\end{example}

\begin{example}
\label{ex:n6_d3_u2}
Let $V = (\F_2)^6$, and let $\circ$ be an operation where $\dim(\wo) = 3$ and $\dim(\uo) = 2$. In this case, the choice of $A$ is more restricted: the error vectors $u_{12}, u_{13},$ and $u_{23}$ are not independent, and their relations must be preserved by $\lb$. For some choices of $A$, the compatibility equations become unsolvable. Let us consider, for example, the operation $\circ$ defined by
\[
\Theta_\circ =
\begin{pmatrix}
000 & 101 & 110 \\
101 & 000 & 101 \\
110 & 101 & 000
\end{pmatrix}.
\]
Here, we have $\b_{12} = \b_{23} = (101)$. Now, take the matrix
\begin{equation*}
A =
\begin{pmatrix}
1 & 1 & 0 \\
1 & 1 & 1 \\
0 & 1 & 0
\end{pmatrix},
\end{equation*}
which is invertible but not a valid choice. Let us show why. Assuming $B = 0$ for simplicity, we get
\[
e_1 \lb = e_1 + e_2, \quad e_2 \lb = e_1 + e_2 + e_3, \quad e_3 \lb = e_2,
\]
which leads to
\[
\b_{12}D = e_1\lb \cdot e_2\lb = (e_1 \cdot e_3) + (e_2 \cdot e_3) = \b_{13} + \b_{23} = (011),
\]
\[
\b_{23}D = \b_{12} + \b_{23} = (000).
\]
However, this is impossible since we had set $\b_{12} = \b_{23}$.

When $\dim(\uo) = 2$, it is more practical to start by choosing $D$. If we fix a basis for $\uo$, we have two vectors whose images must span $\uo$ itself; since $\uo$ contains 3 pairwise independent vectors, there are 6 ways to choose them. We can then extend this to form a basis for $\wo$, and the image of the last vector can be any vector in $\wo \backslash \uo$, giving us 4 options. Thus, there are 24 possible choices for $D$.

The choice of $D$ determines the preimage of $\uo$ through the compatibility equations, which has dimension 2, leaving us with 4 possible choices for $A$, one for each way to map the last vector that does not appear in the equations. Combined with the 512 choices for $B$, we find that $\ho$ has a total of 49152 elements.

\end{example}
We will now show that this is a consequence of a deeper fact, i.e. the conjugacy classes of operations with $d = s-3$ only depends on $\dim(\uo)$.
\begin{theorem}
	\label{thm:coniugio_n3}
	Let $\circ$ and $\diamond$ be two operations such that $\dim(\wo) = \dim(W_\diamond) = s-3$, and let $T_\circ$ and $T_\diamond$ be respectively the associated translation groups. Then, there exists $g \in \GL(V)$ such that $T_\diamond = T_\circ^g$ if and only if $\dim(U_\circ) = \dim(U_\diamond)$.
\end{theorem}
\begin{proof}
Up to conjugation, assume that $W_\circ$ and $W_\diamond$ are generated by $\{e_4,\dots ,e_s\}$, and that $U_\circ$ and $U_\diamond$ are generated by $\{e_4,e_5\}$ if $\dim(U_\circ) = 2$, or by $\{e_4,e_5,e_6\}$ if $\dim(U_\circ) = 3$. The multiplication matrices associated with $e_i$ for $1 \leq i \leq 3$ can be expressed as

\[
\begin{pmatrix}
M^\circ_i & \mathbb{0}_{6,s-6} \\
\mathbb{0}_{s-6,6} & \mathbb{I}_{s-6}
\end{pmatrix}, \quad
\begin{pmatrix}
M^\diamond_i & \mathbb{0}_{6,s-6} \\
\mathbb{0}_{s-6,6} & \mathbb{I}_{s-6}
\end{pmatrix}
\]
respectively. Here, $M^\circ_i$ and $M^\diamond_i$ are $6 \times 6$ matrices because the rows of the multiplication matrices generate $\uo$, as stated in Proposition \ref{prop:info_prod}, and are zero in the last $s-6$ components. We can thus view this sum as consisting of two parallel sums: the first one acting on the first six components and defined by $M^\circ_i$ and $M^\diamond_i$, and the second one acting on the last $s-6$ components as the standard sum.

Now, let $\bar{T}_\circ$ and $\bar{T}_\diamond$ be the translation groups associated with the alternative operations defined by $M^\circ_i$ and $M^\diamond_i$. Based on the classification in Calderini et al.~\cite[Table 1]{calderini2021properties}, we know that for $s=6$ and $d=3$, there are two distinct conjugacy classes of translation groups. Examples \ref{ex:n6_d3_u2} and \ref{ex:n6_d3_u3} demonstrate that these two classes correspond to $\dim(\uo) = 2$ and $\dim(\uo) = 3$. By Lemma \ref{lem:coniugio_ho}, the conjugation on $T_\circ$ induces conjugation on $\ho$, and conjugation must preserve the number of elements. This implies that, if and only if $\dim(\bar{U}_\circ) = \dim(\bar{U}_\diamond)$, there exists $\bar{g} \in \GL_6(V)$ such that $\bar{T}_\diamond = \bar{T}_\circ^{\bar{g}}$. Taking
\[
g =
\begin{pmatrix}
\bar{g} & \mathbb{0}_{6,s-6} \\
\mathbb{0}_{s-6,6} & \mathbb{I}_{s-6}
\end{pmatrix}
\]
we have $T_\diamond = T_\circ^g$ as required. \qed
\end{proof}
\begin{corollary}
    Given two sum $\circ$ and $\diamond$ with $\dim(W_\circ)=\dim(W_\diamond)=s-3$, the corresponding groups $H_\circ$ and $H_\diamond$ are conjugated by $g \in \GL(V)$ if and only if $\dim(U_\circ) = \dim(U_\diamond)$.
\end{corollary}
\begin{proof}
	Apply Lemma \ref{lem:coniugio_ho} to Theorem \ref{thm:coniugio_n3}.\qed
\end{proof}
Comparing Theorem \ref{thm:coniugio_n3} with its analogous Theorem \ref{thm:coniugio_n2} highlights the increased complexity when transitioning to the case where  $d = s-3$. Higher values of  $d$  appear to be even more intricate and are not yet fully understood. We should keep this in mind if we want to  use such sums in a parallel setting; for example, we can no longer assume that the same alternative operation can be applied to all components. Consequently, the techniques employed to prove Theorem \ref{thm:caratt_h0_msomme} become invalid. This, along with the considerations discussed in Section \ref{sec:aut}, clarifies our assumption of  $d = s-2$  for our applications.

\subsection{Considerations on 8-bit s-boxes and the case $d=s-3$}

Our analysis naturally extends beyond the 4-bit setting. In particular, when considering 8-bit s-boxes, experiments show that improvements can indeed be observed, especially for small values of the parameter $2 \leq d \leq 6$. This is expected, as $d$ measures the extent to which the alternative operation resembles the classical xor. For larger s-boxes, however, the phenomenon does not reach the extremal levels observed in dimension four, since constant differentials arise only under very specific algebraic coincidences that are unlikely to occur in higher dimensions. Nevertheless, the evidence shows that the study of alternative operations for larger s-boxes is meaningful and can reveal structural weaknesses that remain invisible to classical xor-based analysis.

Extending our analysis to all 8-bit permutations is not feasible: firstly, because the number of such permutations is extremely large, and secondly, because the number of possible alternative operations is also very high (roughly $2^{47}$~\cite{civino2025classificationsmallbinarybibraces}). Nonetheless, we can analyse  some common s-boxes for the cases $d=5$ and $d=6$ and show that, as expected, the $d=5$ case exhibits better performance in terms of $\circ$-nonlinearity gain (the operation is more dissimilar from xor than in the case $d=6$). We tested all alternative operations in canonical form, that is with $W_\circ$ generated by the last $d$ canonical vectors, (63 for $d=6$ and 32,550 for $d=5$) on the AES~\cite{daemen2002design}, Camellia~\cite{aoki2000camellia}, and Kuznyechik~\cite{biryukov2016reverse} s-boxes, and the results are summarized in Tab.~\ref{tab:common-d6} and Tab.~\ref{tab:common-d5}. These data show that in the $s-3$ setting there is a significant improvement compared to the case $s-2$, justifying a detailed analysis of the $s-3$ case. Although absent in 4-bit instances, it becomes highly relevant for 8-bit s-boxes and exhibits behaviours that deserve attention in their own right. Compared to the $s-2$ case, where the distribution of differential uniformities appears more concentrated, in the $s-3$ case a greater irregularity emerges. This difference can be traced back to the richer structure of the error space $U_\circ$ in the $s-3$ situation, as discussed previously.\\ 

\begin{table}[h!]
\centering
\begin{tabular}{c||c|c|c|}
s-box $\downarrow$ // $\delta^\circ \rightarrow$ & 8 & 10 & 12 \\
\hline\hline
AES        & 55 & 8  & $\cdot$ \\
\hline
Camellia   & 59 & 4  & $\cdot$ \\
\hline
Kuznyechik &   $\cdot$ & 54 & 9 \\
\hline
\end{tabular}
\vspace{2mm} 
\label{tab:common-d6}
\caption{Distribution of $\delta^\circ$-differential uniformities for $d=6$ on selected 8-bit s-boxes.}
\end{table}

\begin{table}[h!]
\centering
\begin{tabular}{c||c|c|c|c|c|c|}
s-box $\downarrow$ // $\delta^\circ \rightarrow$ & 8 & 10 & 12 & 14 & 16 & 18 \\
\hline\hline
AES        & 433 & 23858 & 7841 & 402 & 14 & 2 \\
\hline
Camellia   & 470 & 24087 & 7494 & 476 & 22 & 1 \\
\hline
Kuznyechik & 18 & 18940 & 12425 & 1086 & 80 & 1 \\
\hline
\end{tabular}
\vspace{2mm} 
\label{tab:common-d5}
\caption{Distribution of $\delta^\circ$-differential uniformities for $d=5$ on selected 8-bit s-boxes.}
\end{table}

For the other cases $d=2,3,4$, we have tested 100,000 randomly selected alternative operations on the AES s-box. We remark that, due to the vast number of possible operations, the likelihood of identifying an operation that significantly outperforms the others, if any, is extremely low. The corresponding results are summarized in Tab.~\ref{tab:random-d}.  

\begin{table}[h!]
\centering
\begin{tabular}{c||c|c|c|c|c|c|c|c|}
$d \downarrow$ // $\delta^\circ \rightarrow$ & 8 & 10 & 12 & 14 & 16 & 18 & 20 & 24 \\
\hline\hline
4 & 58 & 54528 & 41515 & 3604 & 281 & 12 & 2 & $\cdot$ \\
\hline
3 & 32 & 52976 & 43111 & 3654 & 216 & 11 & $\cdot$ & $\cdot$ \\
\hline
2 & 267 & 68148 & 29342 & 2073 & 149 & 20 & $\cdot$ & 1 \\
\hline
\end{tabular}
\vspace{2mm} 
\label{tab:random-d}
\caption{Distribution of $\delta^\circ$-differential uniformities for randomly selected operations on the AES s-box for $d=2,3,4$.}
\end{table}

\end{document}